\documentclass[a4paper,twocolumn,11pt]{quantumarticle}
\pdfoutput=1
\usepackage{graphicx} 
\usepackage{amssymb}
\usepackage[utf8]{inputenc}
\usepackage[T1]{fontenc}
\usepackage{amsfonts}
\usepackage{hyperref}
\usepackage{epsfig,pstricks,graphicx}
\usepackage{bm}
\usepackage{physics}
\usepackage{amsthm}
\usepackage{mathtools}
\usepackage[numbers]{natbib}
\usepackage{todonotes}
\def\id{{\rm 1\kern-.22em l}}
\newcommand{\C}{\mathbb{C}}
\newcommand{\proj}[1]{\ketbra{#1}{#1}}
\newcommand{\MM}{\mathcal{M}}
\newcommand{\JJ}{\mathcal{J}}
\newcommand{\PP}{\mathcal{P}}

\newtheorem{propo}{Proposition}

\begin{document}

\title{Single-shot comparison of random quantum channels and measurements}
\author{Marcin Markiewicz}
\affiliation{Institute of Theoretical and Applied Informatics, Polish Academy of Sciences, ul. Ba{\l}tycka 5, 44-100 Gliwice, Poland}
\affiliation{International Centre for Theory of Quantum Technologies, University of Gdansk, 80-309 Gda{\'n}sk, Poland}
\orcid{0000-0002-8983-9077}
\author{{\L}ukasz Pawela}
\affiliation{Institute of Theoretical and Applied Informatics, Polish Academy of Sciences, ul. Ba{\l}tycka 5, 44-100 Gliwice, Poland}
\orcid{0000-0002-0476-7132}
\author{Zbigniew Pucha{\l}a}
\affiliation{Institute of Theoretical and Applied Informatics, Polish Academy of Sciences, ul. Ba{\l}tycka 5, 44-100 Gliwice, Poland}
\orcid{0000-0002-4739-0400}

\begin{abstract}
In this work we provide an efficiency analysis 
of the problem of comparison of two randomly chosen
quantum operations in the single-shot regime. 
We provide tight bounds for the success probability of 
such a protocol for arbitrary quantum channels and 
generalized measurements.
\end{abstract}

\maketitle

\section{Introduction and motivation}

Comparison of unknown quantum operations is a task in which one is given an
ability to apply two given operations with specified input and output dimensions
on arbitrary quantum state, and decide on whether they are identical or distinct
based on final measurement outcomes. 
Our task -- deciding equality of two black‑box channels from a specified ensemble
in a single use per box -- differs from binary discrimination where the two
hypotheses are fully specified \cite{Chefles00,Barnett09}.
Several classes of channels have been
already analyzed in this context (unitary channels
\cite{Hillery10,Soeda21,Hashimoto22}, mixed unitary channels \cite{Sacchi05},
von Neumann measurements \cite{Puchala18,
Puchala2021multipleshot,krawiec2024discrimination}, however there are no known
results regarding comparison of \textit{arbitrary} channels and measurements
in the most general scenario for equality testing of channels drawn at random. Partial results on this issue are
known only for the scenario of unambiguous comparison of general
measurements \cite{Ziman09, Sedlak14}, in which one assumes that the protocol in
each run gives either perfectly correct answer or gives no answer at all. On the
other hand in the scenario discussed in this work one assumes that the algorithm
gives definite answer in  each run, however the answer can be erroneous in some
percentage of runs.
 
Although any completely positive trace preserving (CPTP) channel can be seen
from the \textit{dilation} perspective as a unitary channel on an extended
Hilbert space, the specific optimal performance of  equality testing task
significantly depends on the dimension of the extension and on the way of
sampling the unknown channel, therefore it cannot be directly reduced to the
unitary case itself.

In this work we provide analysis of optimal performance of comparison of
quantum channels and generalized quantum measurements drawn from Haar-Stinespring
ensembles in the single-shot scenario. While channel hypothesis testing studies (adaptive vs non-adaptive)
strategies and asymptotic exponents \cite{Harrow10, Hayashi09, Pirandola19, Zhuang20}, 
we consider the single-use average-case performance of equality testing.

The work is divided into two parts.
Firstly we analyze the performance of comparison of two random channels as a
function of dimensions, which characterize it. In the second part we analyze
comparison of two random POVM-type measurements, based on results on
comparison of arbitrary channels. This can be done since a random POVM
measurement, treated as a channel with classical output, can be represented as a
composition of a random channel with the dephasing channel. Our results provide
the most general solution to the equality testing problem of two quantum
operations in the single-shot regime.

Our main result states that a single, parameter-free strategy is optimal for
all input and output dimensions. The optimal strategy can be summarized as
follows: prepare the antisymmetric two-copy input on the channel ports and
measure with the symmetric projector. No extra ancilla beyond the second port is
needed.

\section{Comparison of unknown channels}\label{sec:randchan}
\subsection{Symmetric comparison}

Let us consider a scenario in which we are given two black boxes, each
containing an unknown quantum channel, with specified input and output
dimensions as respectively $d_i$ and $d_o$. Our goal is to determine whether the
boxes contain the same operation or different ones, and provide a bound on the
probability of correct guessing in this scenario. For a sketch see
Fig.~\ref{fig:discrim}.

\begin{figure}[!h]
\centering\includegraphics[width=\columnwidth]{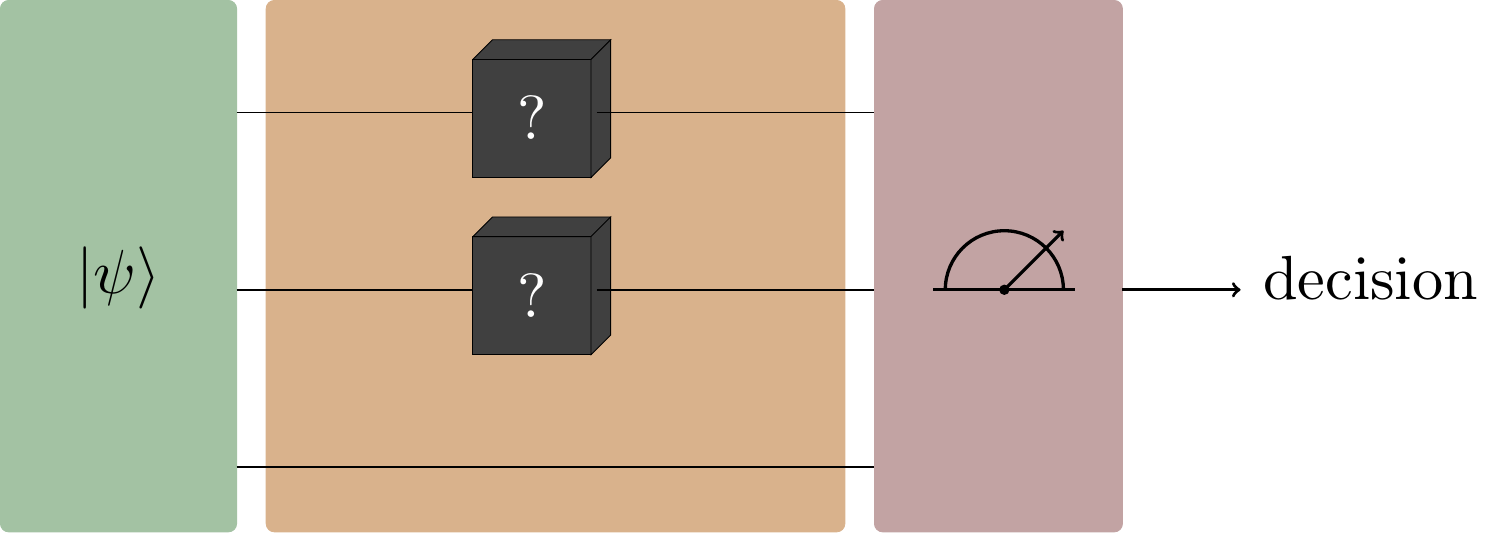} \caption{A
schematic representation of the single-shot symmetric comparison scheme.
Note that the scheme allows for usage of an additional quantum register
untouched by both channels.}\label{fig:discrim}
\end{figure}
A bound on the probability of correct distinction between two channels $\Phi$
and $\Psi$ can be bounded by the diamond norm:
\begin{equation}
    \label{pH}
p^s_H(\Phi, \Psi) \leq \frac{1}{2} + \frac{1}{4} \left\|\Phi - \Psi \right\|_\diamond,
\end{equation}
which can be seen as an old Holevo-Helstrom result \cite{Holevo73, Helstrom69}
restated using Kitaev's diamond norm \cite{Kitaev98, Watrous18}.
Formally, our one-use scheme is a process POVM (tester) \cite{Chiribella08,Gutoski07},
 complementing the network formalism \cite{chiribella2009theoretical}.

Let us denote the set of
$d\times d$ complex matrices by $\mathcal M_d(\mathbb C)$. For a
Hermiticity-preserving $\Xi: \MM_{d_i}(\C) \to \MM_{d_o}(\C)$ with a Choi
representation $J(\Xi) \in \MM_{d_o}(\C) \otimes \MM_{d_i}(\C)$:
$$J(\Xi)=\sum_{i,j=1}^{d_i} \Xi(\ket{i}\!\bra{j})\otimes \ket{i}\!\bra{j}$$
 we can
provide an upper bound on the diamond norm~\cite{nechita2018almost}:
\begin{equation}
    \label{eq:bound-diamond}
\left\| \Xi \right\|_\diamond \leq \left\| \Tr_{d_o} \left| J(\Xi) \right|\right\|.
\end{equation}

Since our task is to find performance of comparison between \textit{unknown}
channels, we have to define a way of sampling such channels. To start, let us
define an arbitrary channel $\Phi_U^{(s)}: \MM_{d_i}(\C) \to \MM_{d_o}(\C)$ via
the Stinespring dilation with extension denoted as $s$:
\begin{equation}
    \label{measurePhi}
    \Phi_U^{(s)}(\rho) = \Tr_{s} U\rho U^\dagger = \Tr_{s} \Phi_U(\rho),
\end{equation}
where the input state $\rho \in \MM_{d_i}(\C)$, $U$ is an isometry operation
such that:
\begin{equation}
    \label{eq:iso}
U: \C^{d_i} \to \C^{d_o} \otimes \C^{s},
\end{equation}
and $\Phi_U(\rho)=U\rho U^{\dagger}$ denotes the corresponding isometry channel.
Hence, a random quantum channel $\Phi_U^{(s)}$ defined in this manner can be
seen as a composition of a random isometry with a partial trace:
\begin{equation}
\Phi_U^{(s)}=\Tr_s\circ \,\Phi_U.
    \label{eq:defPhiU}
\end{equation}
Throughout we fix the input and output dimensions ($d_i$,$d_o$) and an
environment size $s$.

Now we can define the measure $\mu_{\Phi_U^{(s)}}$ on the set of quantum
channels to be derived from to the Haar measure $\mu^s_{d_id_o}$ on the set of
random isometries \eqref{eq:iso}, which can be obtained as a measure on random
truncated unitaries $\textrm{U}(sd_o)$ \cite{Zyczkowski_2000, heinosaari2020random}. Measure
$\mu^s_{d_id_o}$ is a function of three dimensions: $d_i$, $d_o$ and $s$ which
define the isometry \eqref{eq:iso}. For the sake of clarity of further
derivations we will denote the volume differential $\dd\mu^s_{d_id_o}$
corresponding to the measure $\mu^s_{d_id_o}$ as simply $\dd U$.

When considering a quantum channel $\Phi$ via Stinespring representation one
needs to assume  $s d_o \geq d_i$ for the channel  to be trace preserving. At
the same time with the same assumption but different interpretation of $s$ one
can define two another measures on the set of quantum channels
\cite{Kukulski21}, via random Choi matrices and random Krauss operators. As
shown in \cite{Kukulski21} all these three ways of sampling random quantum
channels are equivalent, in a sense that they define the same measure on the set
of channels. For convenience in this work we solely utilize the Stinespring
version.

In the symmetric comparison scenario we are testing equality between
application of random channels $\Phi_U^{(s)}\!\!\otimes\!\Phi_U^{(s)}$
(identical ones) and $\Phi_U^{(s)}\!\!\otimes\!\Phi_V^{(s)}$ (distinct ones),
drawn according to the measure $\mu^s_{d_id_o}$. Such averaged channels are
represented by the integrals $\int\Phi_U^{(s)}\!\!\otimes\!\Phi_U^{(s)}\dd U$
and $\int \Phi_U^{(s)}\!\!\otimes\!\Phi_V^{(s)}\dd U\dd V$ respectively. Hence,
in order to find the Holevo-Helstrom bound \eqref{pH} for the probability
$p^s_H\left(\int\Phi_U^{(s)}\!\!\otimes\!\Phi_U^{(s)}\dd U, \int
\Phi_U^{(s)}\!\!\otimes\!\Phi_V^{(s)}\dd U\dd V\right)$, we need to find the
value of the following expression:
\begin{equation}
\left\| \int\Phi_U^{(s)}\!\!\otimes\!\Phi_U^{(s)}\dd U - \int \Phi_U^{(s)}\!\!\otimes\!\Phi_V^{(s)}\dd U\dd V \right\|_\diamond.
\label{eq:diamPhi}
\end{equation}
In order to provide the diamond norm bounds we need the Choi representation of
both the channels in \eqref{eq:diamPhi}. For now, we will focus on the first
channel. As the only random part of the channel $\Phi_U^{(s)}$
\eqref{eq:defPhiU} is the isometry $\Phi_U$,
from~\cite{krawiec2024discrimination,puchala2017symbolic} we have:
\begin{widetext}
\begin{equation}
\begin{split}
\JJ\left( \int \Phi_U \otimes \Phi_U \dd U \right) &= \frac{1}{(d_o s)^2 - 1}
\left(\id_{(d_o s)^2} \otimes \id_{d_i^2} + S_{(d_o s),(d_o s)} \otimes S_{d_i, d_i}\right) +\\
&-\frac{1}{(d_o s)\left((d_o s)^2 - 1\right)}
\left(\id_{(d_o s)^2} \otimes S_{d_i, d_i} + S_{(d_o s),(d_o s)} \otimes \id_{d_i^2}\right),
\end{split}
\end{equation}
\end{widetext}
where $S_{a,b}$ denotes a swap operation between a $b$-dimensional system and an
$a$-dimensional system. Composing this Choi matrix with the Choi representation
of the partial trace via the notion of link
product~\cite{chiribella2009theoretical} we obtain a simple relation: 
\begin{equation}
\JJ\left(\Tr_s\circ \,\Phi_U\right)=\Tr_s \JJ(\Phi_U), 
\label{PTvsJ}
\end{equation}
which gives, after straightforward yet tedious computations:
\begin{widetext}
\begin{equation}
\begin{split}
\MM_{d_o^2 d_i^2}(\C) \ni \JJ\left(\int \Phi_U^{(s)}\!\!\otimes\!\Phi_U^{(s)}\dd U\right) &= \frac{1}{(d_o s)^2 - 1}
\left( s^2 \id_{d_o^2} \otimes \id_{d_i^2} + s S_{d_o, d_o} \otimes S_{d_i, d_i}\right) +\\
&-\frac{1}{(d_o s)\left((d_o s)^2 - 1\right)}
\left( s^2 \id_{d_o^2} \otimes S_{d_i, d_i} + s S_{d_o, d_o} \otimes \id_{d_i^2}\right).
\end{split}
\end{equation}
\end{widetext}

As for the channel $\int \Phi_U^{(s)}\!\otimes\Phi_V^{(s)}\dd U\dd V$ we follow
the same reasoning as in~\cite{krawiec2024discrimination}:
\begin{equation}
\JJ\left( \int \Phi_U^{(s)}\!\otimes\Phi_V^{(s)}\dd U\dd V \right) = \frac{1}{d_o^2} \id_{d_o^2 d_i^2}.
\label{JUV}
\end{equation}
Finally, in order to apply the bound~\eqref{eq:bound-diamond}, we need to
evaluate the absolute value operator of
\begin{equation}
 J \equiv \JJ\left( \int \Phi_U^{(s)}\!\otimes\Phi_U^{(s)}\dd U \right) - \JJ\left( \int \Phi_U^{(s)}\!\otimes\Phi_V^{(s)}\dd U\dd V \right)\nonumber
\end{equation}
Consider a matrix $W=S_{d_o, d_o} \otimes S_{d_i, d_i}$. As shown in
Appendix~\ref{app:polar} we can construct a polar decomposition of the matrix
$J$ such that we have $W J = |J|$. Hence we have:
\begin{eqnarray}
    \Tr_{d_o} W J &&= \Tr_{d_o} |J|\nonumber\\
     &&=\frac{s(d_o^2 -1)}{d_o^2 s^2 - 1} \id_{d_i^2} - \frac{d_o^2-1}{d_o\left( d_o^2 s^2 - 1 \right)} S_{d_i, d_i}\nonumber\\
     &&= \frac{(d_o^2 - 1)}{d_o(s^2 d_o^2 - 1)}\left( s d_o \id_{d_i^2} - S_{d_i,d_i} \right).\nonumber
\end{eqnarray}
Using the fact that $\|\alpha \id - S\| = 1 + \alpha, \; \alpha > 0$, we finally obtain:
\begin{equation}
\left\|\Tr_{d_o} \left| J \right| \right\| = \frac{d_o^2 - 1}{d_o^2 s - d_o}\label{eq:diamond-value}
\end{equation}
Hence, we get the Holevo-Helstrom bound for symmetric comparison between
random channels:
\begin{eqnarray}
    &&p^s_H\left(\int \Phi_U^{(s)}\!\otimes\Phi_U^{(s)}\dd U, \int \Phi_U^{(s)}\!\otimes\Phi_V^{(s)}\dd U\dd V\right)\nonumber\\
    &&\leq \frac{1}{2} + \frac{1}{4}  \frac{d_o^2 - 1}{d_o^2 s - d_o}
    \label{eq:bound}.
\end{eqnarray}
It turns out, that the above bound for probability of correct distinction can be
saturated using the following input state:
\begin{equation}
    \label{optStateSym}
\C^{d_i} \otimes \C^{d_i} \ni \ket{\psi} = \frac{1}{\sqrt{2}} \left( \ket{01} - \ket{10} \right).
\end{equation}
Appendix~\ref{app:saturation} shows the detailed proof of the
bound~\eqref{eq:bound}. What is noteworthy here is that our bound does not
explicitly depend on the input dimension $d_i$. However, we made an assumption
that $d_i\geq 2$ when proving the saturation of the upper bound. Moreover, note
that due to the fact that the bound is saturated by a bipartite state
\eqref{optStateSym}, there is no need to utilize additional ancillary systems in
the optimal comparison procedure apart from the ones on which the two
channels act (see the lower register in the  Fig. \ref{fig:discrim}). Finally,
let us mention that for $d_o=2$ and $s=1$, namely for the case of testing equality
between two unknown unitary channels, we obtain $p^s_H=\frac{7}{8}$, which
reproduces the famous result by Hillery et. al. \cite{Hillery10}, see
also~\cite{Soeda21}.

In the case of trivial input space, we can think about our channels
$\Phi_U^{(s)}$ and $\Phi_V^{(s)}$ as operations that produce quantum states. In
this case the upper bound is not achieved, however, we can calculate the diamond
norm directly as $\|J\|_1$. This is easily seen as the input space is trivial,
hence no entanglement with the ancillary system can help us achieve greater
separation. We get the following:
\begin{equation}
p^s_H|(d_i=1) = \frac{1}{2} + \frac{1}{4} \|J\|_1 = \frac{1}{2} + \frac{1}{4} \frac{d_o^2 - 1}{d_o\left( d_o s + 1\right)}.
\end{equation}

The final observation is that in the limit of large environment dimension $s \to
\infty$ the success probability of comparison approaches $\frac12$, as the
resulting channels approach the completely depolarizing channel. Yet, if one has
a reasonable control over the environment dimension $s$, then the success probability can be
averaged over the distribution $\pi$ of all accessible 
dimensions of the environment sizes to obtain:
\begin{equation}
\bar{p}_H = \frac12 + \frac14 \frac{d_o^2 - 1}{d_o} \mathbb{E}_{s \sim \pi} \frac{1}{s d_o - 1}.
\end{equation}
A natural example is to assume that the channels are drawn according to the
Hilbert-Schmidt measure on the set of quantum channels
\cite{zyczkowski2001induced}, which implies $\pi(s) = \delta_{s, d_i d_o}$ yielding:
\begin{equation}
\bar{p}_H = \frac12 + \frac14 \frac{d_o^2 - 1}{d_o(d_i d_o^2 - 1)}.
\end{equation}

\subsection{Asymmetric comparison}
Here we focus on an approach different from the Holevo-Helstrom one and move to
a framework based on hypothesis testing, see e.g. \cite{Wilde2020}. We introduce the hypotheses:
\begin{equation}
\begin{split}
    H_0:&\;\mathrm{operations} \; \mathrm{are}\; \mathrm{identical,} \\
    H_1:&\;\mathrm{operations} \; \mathrm{are}\; \mathrm{different.}\nonumber
\end{split}
\end{equation}
Along with these, we introduce two possible types of errors:
\begin{itemize}
\item The operations in the black boxes are the same, however our comparison
scheme stated  that they were different. This is the so-called  type-I error. We
will denote its probability by $p_{\textrm{I}}$.
\item The operations were different, however our comparison scheme stated
that they are identical. This is the so-called type-II error. We will denote its
probability by $p_{\textrm{II}}$. 
\end{itemize}
To estimate the above two errors we assume some initial state $\ket{\psi} \in
\C^{d_i} \otimes \C^{d_i}$, and that we perform a final binary measurement on
the output state of the black boxes with effects $\{\Omega, \id -\Omega\}$. The
outcome associated with $\Omega$ represents accepting the \textit{null
hypothesis}, which states that measurements are the same,  
whereas $\id-\Omega$ indicates acceptance of the \textit{alternative hypothesis}, stating that
they are different. We assume the same channel sampling method as in the
symmetric case. Then the error probabilities in asymmetric comparison read:
\begin{equation}
\label{type12errors}
\begin{split}
p_{\textrm{I}}&=\Tr((\id-\Omega)\rho_{\textrm{id}})=1-\Tr(\Omega\rho_{\textrm{id}}),\\
p_{\textrm{II}}&=\Tr(\Omega\rho_{\textrm{dif}}).
\end{split}
\end{equation}
The output states $\rho_{\textrm{id}}$ and $\rho_{\textrm{dif}}$ represent
outputs of applying respectively the same and different averaged channels  on
the initial state $\ket{\psi}$:
\begin{equation}
\label{eq:averaged-out-states}
\begin{split}
\rho_{\textrm{id}}&=\left(\int \Phi_U^{(s)}\!\otimes\Phi_U^{(s)}\dd U\right)\left(\proj{\psi}\right) =\\
&=\frac{s (d_o s - \bra{\psi} S \ket{\psi})}{d_o(s^2 d_o^2 - 1)} \id_{d_o^2} + \frac{s d_o \bra{\psi} S \ket{\psi}}{d_o(s^2 d_o^2 - 1)} S_{d_o, d_o}\nonumber\\
\rho_{\textrm{dif}}&=\left(\int \Phi_U^{(s)}\!\otimes\Phi_V^{(s)}\dd U\dd V \right)\left(\proj{\psi}\right) = \nonumber\\
&=\frac{1}{d_o^2} \id_{d_o^2}.
\end{split}
\end{equation}
Assuming a completely antisymmetric state, such that $\bra{\psi} S \ket{\psi} =
-1$, we have:
\begin{equation}
\rho_{\textrm{id}} = \frac{1}{d_o (s d_o -1)}\left( s \id_{d_o^2} - S_{d_o, d_o} \right).
\end{equation}
Now, as both $\rho_{\textrm{id}}$ and $\rho_{\textrm{diff}}$ commute with
$S_{d_o, d_o}$ let us take
\begin{equation}
\begin{split}
\Omega &= t_A \Pi_A + t_S \Pi_S, \\
\Pi_A &= \frac12 \left( \id_{d_o^2} - S_{d_o, d_o} \right), \\
\Pi_S &= \frac12 \left( \id_{d_o^2} + S_{d_o, d_o} \right), \\
t_A,\; t_S & \in  [0, 1].
\end{split}
\end{equation}
Simple calculations give us
\begin{equation}
\begin{split}
p_{\textrm{I}} &= 1 - (\alpha t_A + \beta t_S), \\
p_{\textrm{II}} &= c_A t_A + c_S t_S, \\
\alpha &= \frac{(s+1)(d_o - 1)}{2(s d_o - 1)}, \\
\beta &= \frac{(s-1)(d_o + 1)}{2(s d_o - 1)}, \\
c_A &= \frac{d_o - 1}{2 d_o}, \\
c_B &= \frac{d_o + 1}{2 d_o}. \\
\end{split}
\end{equation}
Now we need to solve the following linear program
\begin{equation}
\label{eq:linprog}
\begin{split}
&\min p_{\textrm{II}} \\
\textrm{s.t. } & p_{\textrm{I}} \leq \varepsilon \\
& t_A,\; t_S \in [0, 1].
\end{split} 
\end{equation}
The details of the solution are presented in Appendix~\ref{app:linear-chan}. The optimal
value is
\begin{equation}
\label{pIpIIChannels}
p_{\textrm{II}}^*(\varepsilon) = 
\begin{cases}
\frac{d_o-1}{2 d_o}, & \varepsilon \geq 1-\alpha,\\
\frac{d_o - 1}{2 d_o} + \frac{s d_o - 1}{d_o(s+1)}(1-\varepsilon-\alpha), & \varepsilon < 1 - \alpha.
\end{cases}
\end{equation}

For the special case $s=1$, which corresponds to random unitary channels, we have
$\beta=0$ and $\alpha=1$, $p_{\textrm{I}}$ always equals zero, and we recover the known result
\begin{equation}
\label{pII0channels}
p_{\textrm{II}}^*(0) = \frac{d_o - 1}{2 d_o}.
\end{equation}
In the case of random channels, i.e. $s=d_i d_o$, we have
\begin{equation}
p_{\textrm{II}}^*(\varepsilon) = 
\begin{cases}
\frac{d_o-1}{2 d_o}, & \varepsilon \geq 1-\alpha,\\
\frac{d_o - 1}{2 d_o} + \frac{d_i d_o^2 - 1}{d_o(d_i d_o + 1)}(1-\varepsilon-\alpha), & \varepsilon < 1 - \alpha.
\end{cases}
\end{equation}

It is worth mentioning that there exists an important relation between
 discussed two error probabilities
$p_{\textrm{I}}$ and $p_{\textrm{II}}$ in the asymmetric scheme and the error
probability in the symmetric scheme specified by the Holevo-Helstrom bound,
namely we have:
\begin{equation}
\label{pIpIIineq}
\frac{1}{2}\left(p_{\textrm{I}}+p_{\textrm{II}}\right)\geq 1-p^s_H,
\end{equation}
in which $p^s_H$ is the optimal success probability of a symmetric
comparison protocol. In our case the above inequality is saturated for $\varepsilon=1-\alpha$:
\begin{eqnarray}
\frac12\left(p_{\textrm{I}}+p_{\textrm{II}}\right) &=& \frac12\left(1-\alpha+\frac{d_o-1}{2d_o}\right)\nonumber\\
&=& \frac14\left(2-\frac{d_o^2 - 1}{d_o\left(s d_o - 1\right)}\right),
\end{eqnarray}
which means that such constraint on errors is related with an optimal comparison scheme.
For all other cases (choices of the value of $\varepsilon$) the inequality is sharp.

\section{Comparison of unknown POVMs}
In this section we shall focus on a subset of all quantum channels and limit the
setting shown in Fig.~\ref{fig:discrim} to the case when we are given the
promise that the unknown operations are some POVMs. 
\subsection{Symmetric comparison}
As shown in~\cite{heinosaari2020random} a Haar-random POVM can be defined using
random quantum channels as defined in \eqref{measurePhi}. Let us state here this
correspondence precisely in analogy to the case of  $d$-outcome von Neumann
measurements, which can be seen as compositions of a dephasing map and a random
unitary channel \cite{krawiec2024discrimination}:
\begin{equation}
\PP_{\textrm{vN}}=\Delta\circ\Phi_{\operatorname{U}},
\label{VNchannel}
\end{equation}
in which  the completely dephasing channel is specified as follows:
\begin{equation}
\Delta: \MM_d(\C) \to \MM_d(\C),\; \rho \mapsto \sum_{i=1}^{d} \bra{i}\rho\ket{i} \proj{i}.
\end{equation}
Note that we denote unitary operations by normal font $\operatorname{U}$,
whereas isometries are denoted by italic font as $U$, therefore
$\Phi_{\operatorname{U}}$ represents a unitary channel, whereas $\Phi_U$
represents the isometry channel. The composition \eqref{VNchannel} can be used
to define a measure on the set of von Neumann measurements to be just the Haar
measure  on the unitary group.

Arbitrary POVM measurement with $d_o$ outcomes can be defined as the following
quantum channel with classical output:
\begin{equation}
\PP(\rho)=\sum_{i=1}^{d_o}\Tr(M_i\rho)\proj{i},
 \end{equation}
 in which operators $M_i$ are POVM elements (sometimes called \textit{effects}).
 Following~\cite{heinosaari2020random} we define Haar-random POVM by taking the following effects generated with Haar-random 
 isometries $U$ \eqref{eq:iso}:
 \begin{equation}
    \label{Meffects}
M_i^{U,s}=U^{\dagger}(\proj{i}\otimes \id_s)U.
 \end{equation}
 Then the following holds:
\begin{propo}
A POVM channel of the form:
\begin{equation}
    \PP_U^{(s)}(\rho)=\sum_{i=1}^{d_o}\Tr(M^{U,s}_i\rho)\proj{i},
     \end{equation}
     with effects of the form \eqref{Meffects}, can be represented as:
\begin{equation}
    \label{HaarRandomPOVM}
    \PP_U^{(s)}=\Delta\circ\Phi^{(s)}_U.
\end{equation}
    \end{propo}
    \begin{proof}
Let us rewrite the expression for $\Tr(M^{U,s}_i\rho)$:
\begin{eqnarray}
    \Tr(M^{U,s}_i\rho)&&= \Tr(U^{\dagger}(\proj{i}\otimes \id_s)U\rho)\nonumber\\
    &&= \Tr((\proj{i}\otimes \id_s)U\rho U^{\dagger})\nonumber\\
    &&=\Tr(\sum_{j=1}^s(\proj{i}\otimes \proj{j})U\rho U^{\dagger})\nonumber\\
    &&=\bra{i}\sum_{j=1}^s\bra{j}U\rho U^{\dagger}\ket{j}\ket{i}\nonumber\\
    &&=\bra{i}\Tr_s U\rho U^{\dagger}\ket{i}=\bra{i}\Phi_U^{(s)}(\rho)\ket{i}.\nonumber
\end{eqnarray}
Hence we have:
\begin{eqnarray}
    \PP_U^{(s)}(\rho)&&=\sum_{i=1}^{d_o}\Tr(M^{U,s}_i\rho)\proj{i}\nonumber\\
    &&=\sum_{i=1}^{d_o}\bra{i}\Phi_U^{(s)}(\rho)\ket{i}\proj{i}\nonumber\\
    &&=\Delta(\Phi_U^{(s)}(\rho))=\Delta\circ\Phi_U^{(s)}(\rho).
\end{eqnarray}
\end{proof}
Due to the relation \eqref{HaarRandomPOVM}, we can sample random POVMs according
to the same measure used previously to sample arbitrary random quantum channels,
namely according to the Haar measure on the isometries \eqref{eq:iso}. In full
analogy to the case of symmetric comparison between two unknown random
channels we effectively discriminate between the averaged channels $\int
\PP_U^{(s)} \!\otimes\! \PP_U^{(s)} \dd U$ and $\int \PP_U^{(s)} \!\otimes\!
\PP_V^{(s)} \dd U\dd V$. Based on the above proposition, in order to
characterize symmetric comparison of unknown POVMs via the Holevo-Helstrom
bound for the success probability $p^s_H\left(\int \PP_U^{(s)} \!\otimes\!
\PP_U^{(s)} \dd U, \int \PP_U^{(s)} \!\otimes\! \PP_V^{(s)} \dd U\dd V\right)$,
we need to calculate the diamond norm:
\begin{eqnarray}
&\left\|\int \PP_U^{(s)} \!\otimes\! \PP_U^{(s)} \dd U- \int \PP_U^{(s)} \!\otimes\! \PP_V^{(s)} \dd U\dd V \right\|_\diamond =& \nonumber\\
&\left\| \left( \Delta \otimes \Delta \right) \circ \left( \int \Phi_U^{(s)} \!\otimes\! \Phi_U^{(s)} \dd U- \int \Phi_U^{(s)} \!\otimes\! \Phi_V^{(s)} \dd U\dd V\right) \right\|_\diamond&\nonumber
\end{eqnarray}
Following the same reasoning as in Section~\ref{sec:randchan} we get:
\begin{eqnarray}
&&\JJ\left(\left(\Delta \otimes \Delta \right) \circ \left(\int \Phi_U^{(s)} \!\otimes\! \Phi_U^{(s)} \dd U \right) \right)\nonumber\\
&&= \frac{s}{\left( s d_o \right)^2 - 1} \bigg( s \id_{d_o^2} \otimes \id_{d_i^2} + T_{d_o, d_o} \otimes S_{d_i, d_i}+\nonumber\\
&&- \frac{1}{s d_o} \left( s \id_{d_o^2} \otimes S_{d_i, d_i} + T_{d_o, d_o} \otimes \id_{d_i^2} \right) \bigg).\nonumber
\end{eqnarray}
where $T_{d, d} = \Delta(S_{d, d})$. On the other hand, due to \eqref{JUV}, we
easily obtain:
\begin{equation}
    \int \PP_U^{(s)} \otimes \PP^{(s)}_V  \dd U \dd V = \frac{1}{d_o^2} \id_{d_o^2 d_i^2}.
\end{equation}
Again, we introduce the difference of the Choi representations of the two
discriminated channels:
\begin{equation}
J_{\mathcal P} \equiv \JJ\left(\int \PP_U^{(s)} \!\otimes\! \PP_U^{(s)} \dd U\right) - \JJ\left(\int \PP_U^{(s)} \!\otimes\! \PP_V^{(s)} \dd U\dd V\right) \nonumber
\end{equation}
Using a similar approach to the one described
in~\cite{krawiec2024discrimination} we arrive at:
\begin{equation}
\Tr_{d_o} |J_{\mathcal P}| = \frac{2(d_o - 1)}{d_o(s^2 d_o^2 - 1)}\left( s d_o \id - S \right),
\end{equation}
and a bound for the norm:
\begin{eqnarray}
    &\left\|\int \PP_U^{(s)} \!\otimes\! \PP_U^{(s)} \dd U- \int \PP_U^{(s)} \!\otimes\! \PP_V^{(s)} \dd U\dd V \right\|_\diamond&\nonumber\\
    &\leq \left\| \Tr_{d_o}|J_{\mathcal P}| \right\| = 2\frac{d_o - 1}{d_o\left( s d_o - 1 \right)}.&
\end{eqnarray}
Hence, we get:
\begin{eqnarray}
    &p^s_H\left(\int \PP_U^{(s)} \!\otimes\! \PP_U^{(s)} \dd U, \int \PP_U^{(s)} \!\otimes\! \PP_V^{(s)} \dd U\dd V\right)&\nonumber\\
     &\leq \frac{1}{2} + \frac{1}{2} \frac{d_o - 1}{d_o\left( s d_o - 1
     \right)}.&
     \label{eq:bound-povm}
\end{eqnarray}
Again, the above bound for probability of correct distinction can be saturated
using the maximally antisymmetric input state:
\begin{equation}
\C^{d_i} \otimes \C^{d_i} \ni \ket{\psi} = \frac{1}{\sqrt{2}} \left( \ket{01} - \ket{10} \right).
\end{equation}
The reasoning is similar to that presented in~\cite{krawiec2024discrimination}. Note that for $s=1$, which corresponds 
to the case of comparing von Neumann measurements, we exactly restore the bound
found in \cite{krawiec2024discrimination}.

As in the case of random channels, we can consider the case when one wishes to
average over different environment dimensions $s$ according to some probability
distribution $\pi$. This results in the following averaged success probability:
\begin{equation}
\bar{p}_H = \frac12 + \frac12 \frac{d_o - 1}{d_o} \mathbb{E}_{s \sim \pi} \frac{1}{s d_o - 1}.
\end{equation}
As previously, for the case of Hilbert-Schmidt measure we have $\pi(s) =
\delta_{s, d_i d_o}$ yielding:
\begin{equation}
\bar{p}_H = \frac12 + \frac12 \frac{d_o - 1}{d_o(d_i d_o - 1)}.
\end{equation}

In the scenario of comparison of POVMs the case of trivial input dimension
$d_i=1$ results in a trivial POVM with just one effect.
\subsection{Asymmetric comparison}
Finally, we arrive at the asymmetric comparison of unknown POVMs. Following the same line of thinking as in Section~\ref{sec:randchan}
we get:
\begin{eqnarray}
&&\rho_\textrm{id}=\frac{1}{d_o\left(s^2 d_o^2 - 1 \right)} \bigg( s\big(s d_o - \bra{\psi} S \ket{\psi}\big) \id + \nonumber\\
&&+\big(s d_o \bra{\psi} S \ket{\psi} - 1\big)T \bigg) \nonumber\\
&&\rho_{\textrm{dif}}=\frac{1}{d_o^2} \id_{d_o^2}.
\end{eqnarray}
Setting once again $\bra{\psi} S \ket{\psi} = -1$, which corresponds to taking completely antisymmetric state $\ket{\psi}$, we have:
\begin{equation}
\rho_\textrm{id} =\frac{1}{d_o\left(s d_o - 1 \right)} \left( s \id -  T \right).
\end{equation}
As in the case of random channels, we choose the final measurement effects to be
of the form:
\begin{equation}
\begin{split}
\Omega &= t_A \Pi_A + t_S \Pi_S, \\
\Pi_A &= \left( \id_{d_o^2} - T \right), \\
\Pi_S &= T, \\
t_A,\; t_S & \in  [0, 1].
\end{split}
\end{equation}
Following the similar reasoning as in the case of random channels we arrive at:
\begin{equation}
\begin{split}
p_{\textrm{I}} &= 1 - (\alpha t_A + \beta t_S), \\
p_{\textrm{II}} &= c_A t_A + c_S t_S, \\
\alpha &= \frac{s(d_o - 1)}{s d_o - 1}, \\
\beta &= \frac{(s-1)}{ d_o - 1}, \\
c_A &= \frac{d_o - 1}{d_o}, \\
c_B &= \frac{1}{d_o}. \\
\end{split}
\end{equation}
Solving the same problem as in \eqref{eq:linprog} we get the optimal value:
\begin{equation}
    \label{pIpIIPOVMS}
p_{\textrm{II}}^*(\varepsilon) = 
\begin{cases}
    \frac{s d_o - 1}{s d_o}(1 - \varepsilon)), & \varepsilon \geq 1 - \alpha, \\
1 - \frac{s d_o - 1}{d_o(s - 1)} \varepsilon, & \varepsilon < 1 - \alpha.
\end{cases}
\end{equation}
The details of the solution are presented in Appendix~\ref{app:linear-povm}.

For the special case $s=1$, i.e. random von Neumann measurements, we have
$\beta=0$ and $\alpha=1$, and setting $\varepsilon=0$ we recover the known result~\cite{krawiec2024discrimination}:
\begin{equation}
    \label{p0POVMS}
p_{\textrm{II}}^*(0) = \left( 1 - \frac{1}{d_o} \right).
\end{equation}

In full analogy with the case of asymmetric comparison of random channels
the inequality \eqref{pIpIIineq} is saturated for $\varepsilon=1-\alpha$, for which we have:
\begin{equation}
\frac{1}{2}( p_{\textrm{I}}+ p_{\textrm{II}})=\frac{d_o^2s-2d_o+1}{2d_o(d_os-1)}.
\end{equation}
This expression exactly equals:
\begin{equation}
1-p^s_H=\frac{1}{2} - \frac{1}{2} \frac{d_o - 1}{d_o\left( s d_o - 1 \right)}=\frac{d_o^2s-2d_o+1}{2d_o(d_os-1)},
\end{equation}
in which we take for $p^s_H$ the upper bound from \eqref{eq:bound-povm}.

\section{Discussion}
In this work we have analyzed the problem of comparing two unknown
arbitrary quantum channels and arbitrary POVM measurements drawn according to
their corresponding Haar measures using both symmetric and asymmetric
comparison schemes in the single application regime.
In particular, our findings show that the comparison task admits a remarkably simple optimal solution. A single, fully parameter-free strategy achieves the minimum possible error rates for all input and output dimensions. The procedure is straightforward: one prepares the antisymmetric two-copy input state on the channel ports and performs a measurement that projects onto the symmetric and antisymmetric subspaces. No auxiliary systems beyond the second port are required. This universality and structural simplicity highlight that optimal performance in the comparison problem does not rely on sophisticated entangled resources, but rather emerges directly from the symmetry properties of the task.
\subsection{Random Channels}
Let us first discuss the case of comparison of random channels. For the case
of symmetric comparison the \textit{optimal} comparison probability
$p^s_H$ reads:
\begin{eqnarray}
    &&p^s_H\left(\int \Phi_U^{(s)}\!\otimes\Phi_U^{(s)}\dd U, \int \Phi_U^{(s)}\!\otimes\Phi_V^{(s)}\dd U\dd V\right)\nonumber\\
    &&= \frac{1}{2} + \frac{1}{4}  \frac{d_o^2 - 1}{d_o^2 s - d_o},
\end{eqnarray}
and is depicted in the Figure \eqref{Fig:Channels} as a function of the output
dimension treated for clarity as a continuous parameter, for three different
values of the environment dimension $s$.
\begin{figure}[!h]
    \centering\includegraphics[width=\columnwidth]{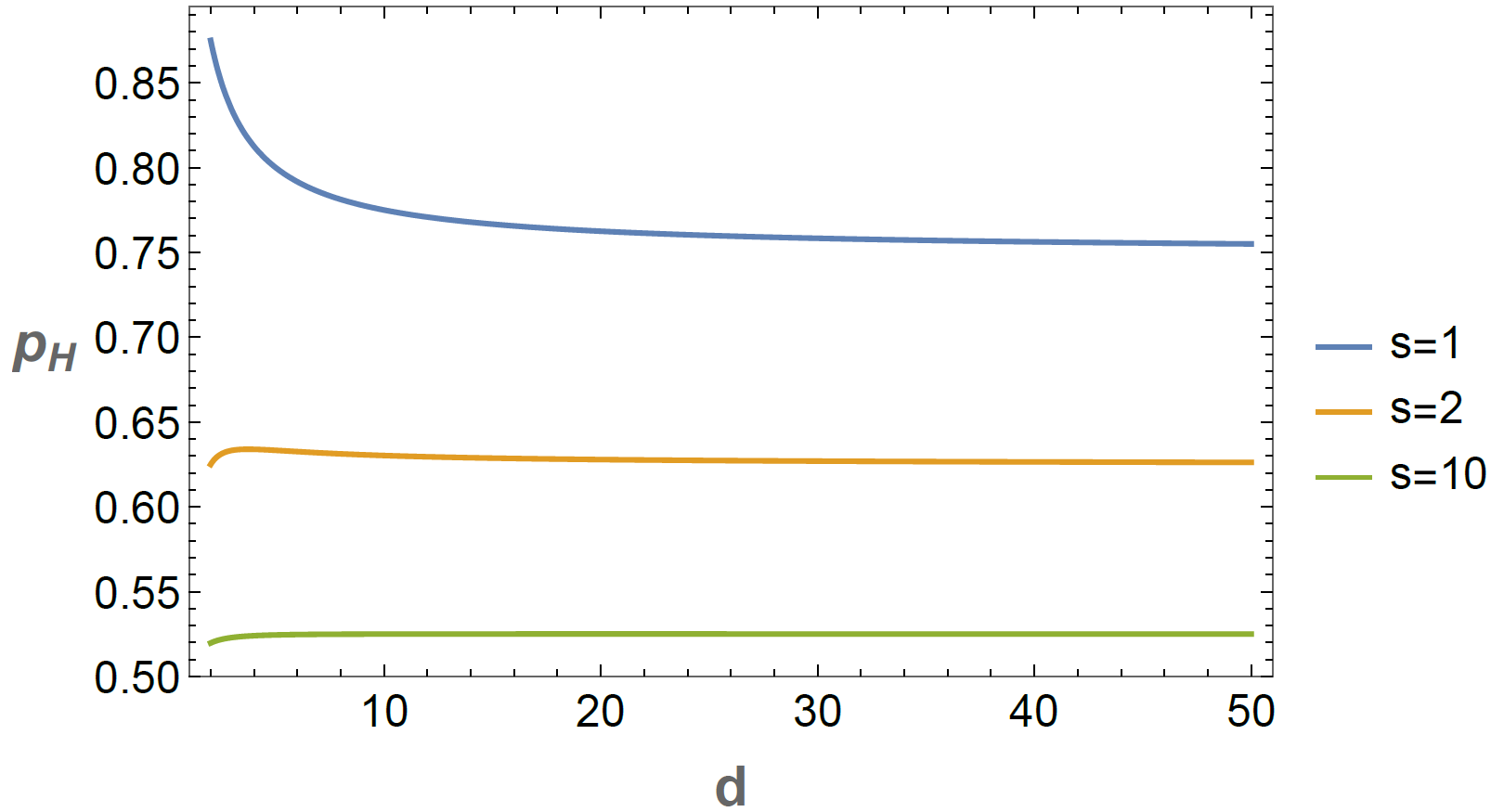}
        \caption{ Plot of the optimal value of a success probability for
        comparison of two random channels in a symmetric scheme as a
        function of (continuated) output dimension $d$, presented for three
        different values of the environment dimension $s$. The case $s=1$ corresponds to
        random unitary channels.}\label{Fig:Channels}
    \end{figure}
It is worth mentioning that for the case of comparison of random unitary
channels ($s=1$), the optimal success probability is maximal for the output
dimension equal $d_o=2$, whereas in the case of random non-unitary channels the
maximum is shifted towards higher values of the output dimension $d_o$. The
asymptotic values for $p^s_H$ in the limit of $d_o\rightarrow\infty$ are
depicted in the Figure \ref{Fig:ChannelsAsympt}, and they tend to $\tfrac{1}{2}$
in the limit of large $s$.
    \begin{figure}[!h]
        \centering\includegraphics[width=\columnwidth]{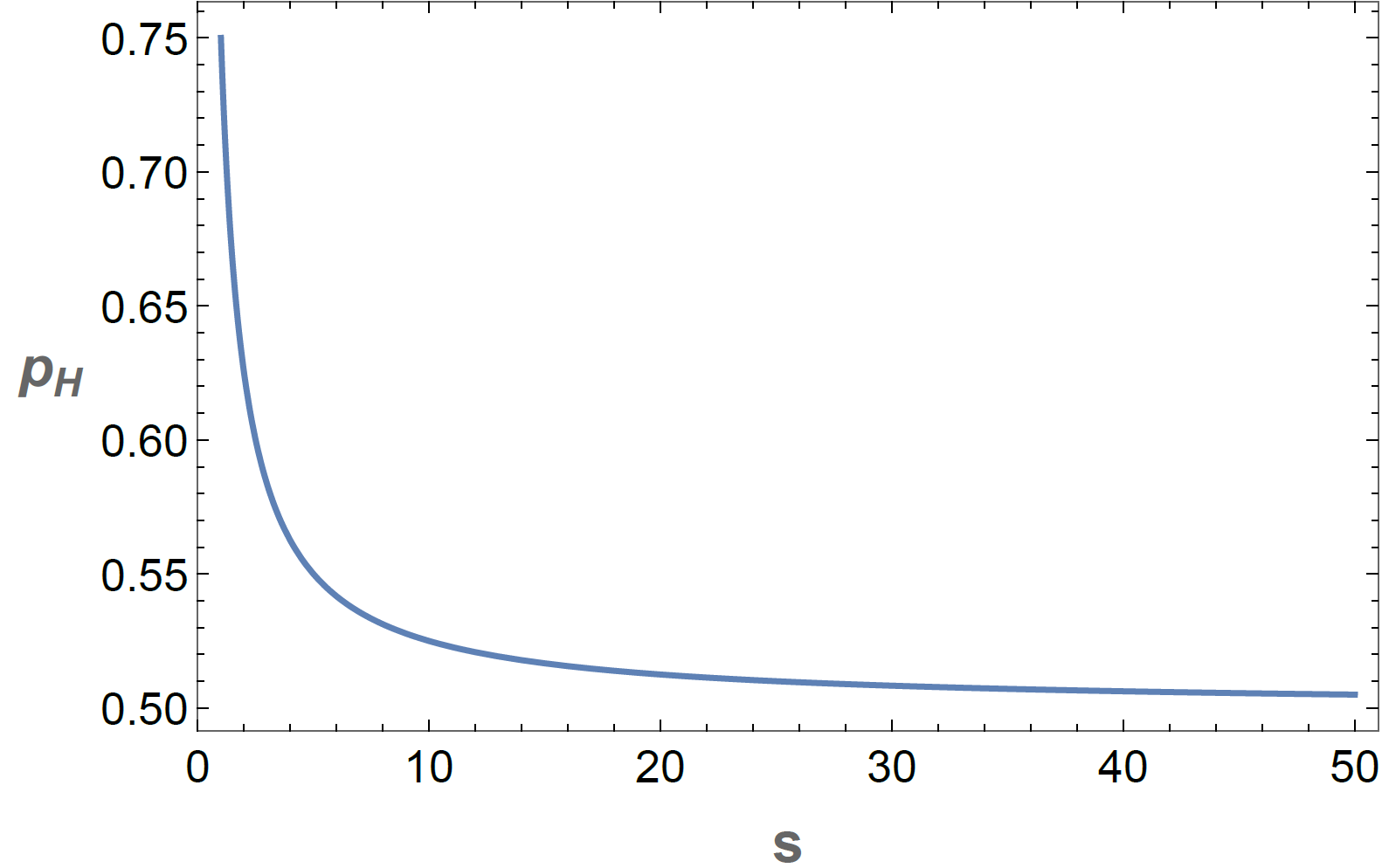}
            \caption{ Plot of the value of a success probability for
            comparison of two random channels in a symmetric scheme as a
            function of (continuated) environment dimension $s$, taken in the
            asymptotic limit of infinite output
            dimension.}\label{Fig:ChannelsAsympt}
        \end{figure}

In the case of asymmetric comparison of random channels
\eqref{pIpIIChannels}, note that the minimal type-I error probability is equal
to zero only for a comparison of two unitary channels ($s=1$), see \eqref{pII0channels}. This means
that only for unitary channels one can completely avoid the type-I error, namely
the case of erroneously deciding that identical channels are different. In the
case of non-unitary channels generated with environment of arbitrary dimension
$s\geq 2$ such a situation is unavoidable. 

\subsection{Random POVMs}
The \textit{optimal} probabilities for a comparison of two Haar-random POVMs
in the symmetric comparison scheme:
\begin{eqnarray}
    &p^s_H\left(\int \PP_U^{(s)} \!\otimes\! \PP_U^{(s)} \dd U, \int \PP_U^{(s)} \!\otimes\! \PP_V^{(s)} \dd U\dd V\right)&\nonumber\\
     &= \frac{1}{2} + \frac{1}{2} \frac{d_o - 1}{d_o\left( s d_o - 1 \right)}.&
     \label{eq:bound-povm-disc}
\end{eqnarray}
behave quite different from those for Haar-random quantum channels, see Figure
\ref{Fig:POVMS}. Namely, in this case the success probability attains its maximal
value always for the output dimension $d_o=2$, as is the case for comparing
 two random von Neumann measurements ($s=1$). Also in the asymptotic
limit of large output dimension, $p^s_H$ tends to $\tfrac{1}{2}$ for all
environment dimensions $s$.
\begin{figure}[!h]
    \centering\includegraphics[width=\columnwidth]{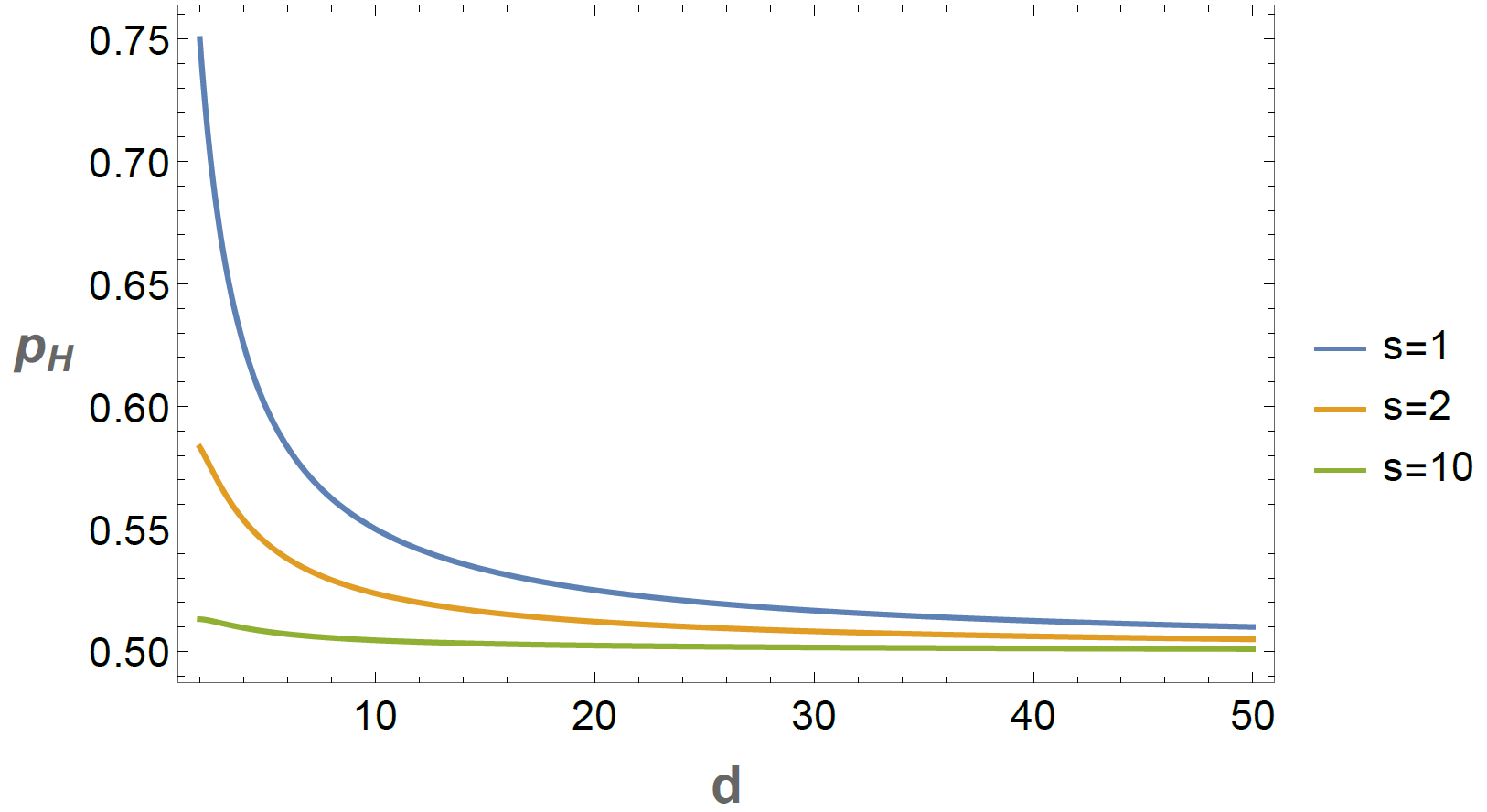} \caption{
        Plot of the optimal value of a success probability for comparison
        of two random POVM measurements in a symmetric scheme as a function of (continuated) 
        output dimension $d$, presented for three different values of the environment dimension $s$.
        The case $s=1$ corresponds to random von Neumann measurements.}\label{Fig:POVMS}
    \end{figure}

In the case of asymmetric comparison scheme we observe analogous situation
like in comparison of random channels, namely the minimal value of the
type-I error probability \eqref{pIpIIPOVMS} attains zero only in the case of von
Neumann measurements ($s=1$) \eqref{p0POVMS}. 
Our conclusions regarding the ($s=1$) case for both unitary channels and 
von Neumann measurements connect to perfect/near-perfect discrimination of unitary operations \cite{Duan07,Duan09,Acin2001}.
Therefore, in the case of comparison of POVM
measurements generated with environment of arbitrary dimension $s\geq 2$,
type-I error corresponding to the case of erroneously treating identical
measurements as distinct ones is unavoidable.

 Note that in this work we focused on providing the most general solution to the
 problem of single-shot comparison of two arbitrary quantum channels and
 measurements. A natural continuation of this line of research would be to
 consider multiple-uses scenario for comparison of general channels.

\section*{Acknowledgements}
MM acknowledges support from the National Science Center (NCN), Poland, under
Project Opus No. 2024/53/B/ST2/02026. ZP and ŁP acknowledge support from the
National Science Center (NCN), Poland, under Project Opus No.
2022/47/B/ST6/02380.

\bibliographystyle{quantum}

\appendix
\section{Polar decomposition of the Choi matrix for general channels}\label{app:polar}
\begin{widetext}
Let
\begin{equation}
J = \frac{1}{(d_o s)^2 - 1} \left( s^2 \id_{d_o^2} \otimes \id_{d_i^2} + s S_{d_o, d_o} \otimes S_{d_i, d_i}\right)
-\frac{1}{(d_o s)\left((d_o s)^2 - 1\right)}\left( s^2 \id_{d_o^2} \otimes S_{d_i, d_i} + s S_{d_o, d_o} \otimes \id_{d_i^2}\right)
- \frac{1}{d_o^2} \id_{d_o^2} \otimes \id_{d_i^2},
\end{equation}
and
\begin{equation}
W = S_{d_o, d_o} \otimes S_{d_i, d_i}.
\end{equation}
In order to show that $J = W |J|$ is a proper polar decomposition, it suffices to show that $J^2 = (WJ)^2$.
We have:
\begin{eqnarray}
WJ &&= \frac{s}{(d_o s)^2 - 1} \bigg( s S_{d_o, d_o} \otimes S_{d_i, d_i} + \id_{d_o^2} \otimes \id_{d_i^2} - \frac{1}{d_o s}\left( s S_{d_o, d_o} \otimes \id_{d_i^2} + \id_{d_o^2} \otimes S_{d_i, d_i} \right) \bigg),
\end{eqnarray}
and
\begin{equation}
\begin{split}
(WJ)^2 &= \left( \frac{s}{(d_o s)^2 - 1} \right)^2 \left( s S_{d_o, d_o} \otimes S_{d_i, d_i} + \id_{d_o^2} \otimes \id_{d_i^2} - \frac{1}{d_o s}\left( s S_{d_o, d_o} \otimes \id_{d_i^2} + \id_{d_o^2} \otimes S_{d_i, d_i} \right) \right)^2 +\\
&- \frac{2s}{d_o^2 ((d_o s)^2 - 1)} \left( s \id_{d_o^2} \otimes \id_{d_i^2} + S_{d_o, d_o} \otimes S_{d_i, d_i} - \frac{1}{d_o s}\left( s\id_{d_o^2} \otimes S_{d_i, d_i} + S_{d_o, d_o} \otimes \id_{d_i^2} \right) \right) +\\
&+ \frac{1}{d_o^4} \id_{d_o^2} \otimes \id_{d_i^2}.\label{eq:WJ2}
\end{split}
\end{equation}
On the other hand, we have:
    \begin{equation}
    \begin{split}
    J^2 &= \left( \frac{s}{(d_o s)^2 - 1} \right)^2 \left( s\id_{d_o^2} \otimes \id_{d_i^2} + S_{d_o, d_o} \otimes S_{d_i, d_i} - \frac{1}{d_o s}\left( s \id_{d_o^2} \otimes S_{d_i, d_i} + S_{d_o, d_o} \otimes \id_{d_i^2} \right) \right)^2 +\\
    &- \frac{2s}{d_o^2 ((d_o s)^2 - 1)} \left( s \id_{d_o^2} \otimes \id_{d_i^2} + S_{d_o, d_o} \otimes S_{d_i, d_i} - \frac{1}{d_o s}\left( s\id_{d_o^2} \otimes S_{d_i, d_i} + S_{d_o, d_o} \otimes \id_{d_i^2} \right) \right) +\\
    &+ \frac{1}{d_o^4} \id_{d_o^2} \otimes \id_{d_i^2}.\label{eq:J2}
    \end{split}
    \end{equation}
\end{widetext}
What is left is to show that the first terms in the expressions \eqref{eq:WJ2}
and \eqref{eq:J2} are equal. This can be done by direct calculation, which we
omit here for brevity.
\section{Saturation of diamond norm for general channels}\label{app:saturation}
Let
\begin{equation}
\begin{split}
J &= \frac{1}{(d_o s)^2 - 1} \left( s^2 \id_{d_o^2} \otimes \id_{d_i^2} + s S_{d_o, d_o} \otimes S_{d_i, d_i}\right) +\nonumber\\
& -\frac{1}{(d_o s)\left((d_o s)^2 - 1\right)}\left( s^2 \id_{d_o^2} \otimes S_{d_i, d_i} + s S_{d_o, d_o} \otimes \id_{d_i^2}\right) +\nonumber\\
& - \frac{1}{d_o^2} \id_{d_o^2} \otimes \id_{d_i^2},
\end{split}
\end{equation}
and
\begin{equation}
    \rho = \frac{1}{2} \left( \ket{01} - \ket{10} \right)\left( \bra{01} - \bra{10} \right)
\end{equation}
We calculate
\begin{equation}
    \Tr_{d_i} J\left(\id_{d_o^2} \otimes \rho^\top \right) = \frac{1}{d_o \left( d_o s-1 \right)}\left(\frac{1}{d} \id_{d_o^2} - S_{d_o, d_o} \right).
\end{equation}
Note that
\begin{equation}
\|\alpha \id_{d_o^2} - S_{d_o, d_o} \|_1 = d_o (d_o - \alpha).
\end{equation}
Thus
\begin{equation}
\left\| \Tr_{d_i} J \left(\id_{d_o^2} \otimes \rho^\top \right)\right\|_1 = \frac{d_o^2 - 1}{d_o(d_o s - 1)}
\end{equation}
which recovers the value in~\eqref{eq:diamond-value}.

\section{Solution to the linear programming problem for asymmetric
comparison of random channels}\label{app:linear-chan}

Recall the states to be discriminated:
\begin{align}
\rho_{\mathrm{id}}
&= \frac{1}{d_o (s d_o - 1)}\big( s\,\id - S \big),\\
\rho_{\mathrm{dif}}
&= \frac{1}{d_o^{2}}\,\id,
\end{align}
Given an effect $0\le \Omega\le \id$ the type-I and type-II 
error probabilities read:
\begin{equation}
\begin{split}
p_{\textrm{I}}(\Omega) &= 1 - \Tr(\Omega \rho_{\mathrm{id}}),\\
p_{\textrm{II}}(\Omega) &= \Tr(\Omega \rho_{\mathrm{dif}}).
\end{split}
\end{equation}
The asymmetric hypothesis testing problem is
\begin{equation}
\min_{0\le \Omega\le \id} p_{\textrm{II}}(\Omega)
\quad \text{s.t.} \quad
p_{\textrm{I}}(\Omega)\le \varepsilon.
\end{equation}

\paragraph{Reduction to two scalars.}
Because $\rho_{\mathrm{id}}$ and $\rho_{\mathrm{dif}}$ lie in the commutative
algebra $\mathrm{span}\{\id, S\}$, we may restrict to effects diagonal in
the swap eigenspaces:
\begin{equation}
\Omega = t_A\,\Pi_A + t_S\,\Pi_S, \qquad 0\le t_A,t_S\le 1,
\end{equation}
where the projectors onto the antisymmetric and symmetric subspaces are
\begin{equation}
\Pi_A = \frac{1}{2}(\id-S),\qquad
\Pi_S = \frac{1}{2}(\id+S).
\end{equation}
Write the resulting linear forms
\begin{equation}
p_{\textrm{I}} = 1 - (\alpha\, t_A + \beta\, t_S),\qquad
p_{\textrm{II}} = c_A\, t_A + c_S\, t_S,
\end{equation}
where the coefficients are
\begin{equation}
\begin{split}
\alpha &= \Tr(\Pi_A \rho_{\mathrm{id}}) = \frac{(s+1)(d_o-1)}{2\,(s d_o - 1)},\\
\beta &= \Tr(\Pi_S \rho_{\mathrm{id}}) = \frac{(s-1)(d_o+1)}{2\,(s d_o - 1)},\\
c_A &= \Tr(\Pi_A \rho_{\mathrm{dif}}) = \frac{d_o-1}{2 d_o}, \\
c_S &= \Tr(\Pi_S \rho_{\mathrm{dif}}) = \frac{d_o+1}{2 d_o}.
\end{split}
\end{equation}

\paragraph{Linear-program geometry.}
Thus, we have the following linear program:
\begin{equation}
\begin{split}
&\min_{t_A, t_S} c_A t_A + c_S t_S\\
\textrm{s.t. }& \alpha t_A + \beta t_S \ge 1 - \varepsilon,\\
&0\le t_A,t_S\le 1.
\end{split}
\end{equation}

With the parametrization above, the constraint set is a rectangle intersected
with the half-space $\alpha t_A + \beta t_S \ge 1-\varepsilon$. Because
$\frac{c_S}{\beta}>\frac{c_A}{\alpha}$, the objective slope is steeper along the
$S$-direction while the constraint grows faster in the $A$-direction, hence the
minimum is attained by first setting $t_A=1$. If $\alpha \geq 1-\varepsilon$,
then $(t_A^*,t_S^*)=(1,0)$ is feasible and optimal, yielding $p_{\textrm{II}}^{\star}=c_A$.
Otherwise, set $t_A=1$ and choose the minimal $t_S$ that satisfies the
constraint, $t_S=(1-\varepsilon-\alpha)/\beta$, which gives $p_{\textrm{II}}^{\star}=c_A
+ c_S\,\frac{1-\varepsilon-\alpha}{\beta}$. 

Putting everything together, we have:
\begin{equation}
p_{\textrm{II}}^{\star}(\varepsilon)=
\begin{cases}
\frac{d_o-1}{2 d_o}, & \varepsilon \geq 1-\alpha,\\
\frac{d_o-1}{2 d_o} + \frac{s d_o - 1}{d_o (s-1)} \left(1-\varepsilon-\alpha\right), & \varepsilon < 1-\alpha,
\end{cases}
\end{equation}
and the optimizer is
\begin{equation}
(t_A^{\star},t_S^{\star})=
\begin{cases}
(1,0), & \varepsilon \geq 1-\alpha,\\
\left(1,\,\frac{1-\varepsilon-\alpha}{\beta}\right), & \varepsilon < 1-\alpha.
\end{cases}
\end{equation}

\section{Solution to the linear programming problem for asymmetric
comparison of random POVMs}\label{app:linear-povm}
\paragraph{Reduction to two scalars.} Because $\rho_{\mathrm{id}}$ and $\rho_{\mathrm{dif}}$ now lie in the commutative
algebra $\mathrm{span}\{\id, T\}$, with $T=\Delta(S)$ being a dephased swap operator, we may restrict to effects diagonal in
the $T$ eigenspaces:
\begin{equation}
\Omega = t_A\,\Pi_A + t_S\,\Pi_S, \quad 0\le t_A,t_S\le 1,
\end{equation}
where the projectors are
\begin{equation}
\Pi_A = (\id-T),\quad
\Pi_S = T.
\end{equation}
Using these definitions we have:
\begin{equation}
\begin{split}
    \alpha & = \frac{s(d_o - 1)}{s d_o - 1},\\
    \beta & = \frac{(s-1)}{ sd_o - 1}, \\
    c_A & = \frac{d_o - 1}{d_o}, \\
    c_S & = \frac{1}{d_o}. \\
\end{split}
\end{equation}
Hence, we have the following linear program:
\begin{equation}
\begin{split}
& \min_{t_A,t_S} c_A t_A+c_S t_S \\
\textrm{s.t. } & \alpha t_A+\beta t_S \geq 1-\varepsilon, \\
& \quad 0 \leq t_A, t_S \leq 1.
\end{split}
\end{equation}
Again, we notice that $\frac{c_S}{\beta}>\frac{c_A}{\alpha}$, hence we increase
$t_A$ before $t_S$. Following the same reasoning as in Appendix~\ref{app:linear-chan}, we first increase
$t_A$ and leave $t_S=0$. The optimal value is attained when the constraint is
saturated, i.e. $\alpha t_A = 1-\varepsilon$ when $\varepsilon \geq 1-\alpha$. 

Now, set $t_A=1$ and choose $t_S$ such that the constraint is saturated when
$\varepsilon < 1-\alpha$, i.e. $t_S = \frac{1-\alpha - \varepsilon}{1-\alpha}$.

Combining everything we get
\begin{equation}
p_{\textrm{II}}^{\star}(\varepsilon)=
\begin{cases}
\frac{s d_o-1}{s d_o}\,(1-\varepsilon), & \varepsilon \geq 1 -\alpha,\\
1-\frac{s d_o-1}{d_o(s-1)} \varepsilon, & \varepsilon < 1 -\alpha,
\end{cases}
\end{equation}
and the optimizer is
\begin{equation}
(t_A^{\star},t_S^{\star})=
\begin{cases}
( \frac{1-\varepsilon}{\alpha}, 0), & \varepsilon \geq 1 -\alpha,\\
(1, \frac{1-\alpha - \varepsilon}{1-\alpha}), & \varepsilon < 1 -\alpha.
\end{cases}
\end{equation}

For the von Neumann case, $s=1$ we have $\beta=0$, $\alpha=1$. The constraint is
simply $t_A \geq 1-\varepsilon$ and the optimizer is
$t_A^{\star}=1-\varepsilon$, $t_S^{\star}=0$, so
\begin{equation}
p_{\textrm{II}}^{\star}(\varepsilon)=\left(1-\frac{1}{d_o}\right)(1-\varepsilon).
\end{equation}

\end{document}